\documentclass[12pt,intlimits]{amsart}
\usepackage[a4paper,top=2cm,bottom=2cm,left=2cm,right=2cm,bindingoffset=5mm]{geometry}
\usepackage{amsmath}
\usepackage{amsthm}
\usepackage{amssymb}
\usepackage{cite}
\usepackage{tikz}
\usepackage{float}

\usepackage{graphicx}
\usepackage{epsfig}
\usepackage{hyperref}
\hypersetup{colorlinks=true,linkcolor=blue}
\newtheorem{theo}{Theorem}[section]
\newtheorem{lem}[theo]{Lemma}
\newtheorem{defi}[theo]{Definition}
\newtheorem{notation}[theo]{Notation}
\newtheorem{folg}[theo]{Corollary}
\newtheorem{prop}[theo]{Proposition}
\newtheorem{remark}[theo]{Remark}

\newcommand{\be}{\begin{eqnarray}}
\newcommand{\ee}{\end{eqnarray}}
\newcommand{\bes}{\begin{eqnarray*}}
\newcommand{\ees}{\end{eqnarray*}}
\newcommand{\bi}{\begin{itemize}}
\newcommand{\ei}{\end{itemize}}
\newcommand{\ben}{\begin{enumerate}}
\newcommand{\een}{\end{enumerate}}

\newcommand{\R}{\mathbb{R}}
\newcommand{\N}{\mathbb{N}}

\newcommand{\de}{\mathrm {d}}

\def\Ext{{\hbox{\rm Ext}}}

\def\einschr{\hbox{\kern1pt\vrule height 6pt\vrule  width6pt height 0.4pt depth0pt\kern1pt}}

\DeclareMathOperator{\dive}{div}
\DeclareMathOperator{\supp}{supp}
\DeclareMathOperator{\curl}{curl}
\DeclareMathOperator{\Ker}{Ker}
\DeclareMathOperator{\capac}{cap}
\DeclareMathOperator{\diam}{diam}

\DeclareMathOperator{\vek}{vec}

\title{\bf  Magnetostatic problems in fractal domains  }

\date{}

\begin{document}
\maketitle

\centerline{\scshape Simone Creo, Maria Rosaria Lancia and Paola Vernole}
\medskip
{\footnotesize

 \centerline{Dipartimento di Scienze di Base e Applicate per l'Ingegneria, Universit\`{a} degli studi di Roma Sapienza,
}
   \centerline{Via A. Scarpa 16,}
   \centerline{00161 Roma, Italy.}
} 

\medskip

\centerline{\scshape Michael Hinz}
\medskip
{\footnotesize
 \centerline{Department of Mathematics, Bielefeld University,}
 \centerline{Postfach 100131,}
   \centerline{33501 Bielefeld, Germany.}
}

\medskip

\centerline{\scshape Alexander Teplyaev}
\medskip
{\footnotesize
 \centerline{Department of Mathematics, University of Connecticut,}
 \centerline{341 Mansfield Road U1009,}
   \centerline{06269-1009 Storrs, Connecticut, USA.}
}

\bigskip

\begin{abstract}
\noindent We consider a magnetostatic problem in a 3D ``cylindrical" domain of Koch type. We prove existence and uniqueness results for both the fractal and pre-fractal problems and we investigate the convergence of the pre-fractal solutions to the limit fractal one. We consider the numerical approximation of the pre-fractal problems via FEM and we  {give} a priori error estimates. Some numerical simulations are also shown. Our long term motivation includes studying problems that appear in quantum physics in fractal domains. 

\end{abstract}
\medskip


\tableofcontents

\newpage
\section{Introduction}

The aim of this paper is to study a magnetostatic problem in a fractal domain. Trying to understand the magnetic properties of fractal structures is a new challenge from both the practical and theoretical point of view. In general mathematical physics on fractals is still a young subject, see \cite{A, AkkermansMallick, ADT, ADT2010PRL, ADTV} for some results; magnetic operators on fractal spaces have been studied only very recently, \cite{H14b, HR14, HTc, HKMRS17}, as well as heat transfer across fractal layers or boundaries \cite{La-Ve1,JEE,JMAA,LVSV,CPAA,nostronazarov,AR-P17,R-PGS}. 
Our long term motivation includes a possibility to study   
non-quantized penetration of magnetic field in the vortex state of superconductors
\cite{geim2000non} 
in fractal domains. 

A mathematical theory of electrodynamics on domains with fractal boundary still has to be developed. Although many results are well known in the case of Lipschitz domains, see for instance \cite[Chapter IX]{DautrayLions}, for such fractal domains even the simplest models and effects have not yet been discussed. Our considerations here should be regarded as a preliminary step in a long term project, which aims to provide theoretical and numerical studies of related physical phenomena. We believe that, beyond their theoretical interest, such results may also be useful for the construction of concrete prototypes in industrial applications, which aim to maximize (or minimize) physical quantities such as the intensity of the magnetic vector field induced by a given current density.

In the present paper we consider a linear magnetostatic problem in a cylindrical three-dimensional domain $Q=\Omega\times I$, where $\Omega$ is the two-dimensional snowflake domain with Koch-type boundary $F$ and $I$ is the unit interval. We consider the problem of finding a divergence free magnetic vector potential for given time-independent permeability and time-independent current density, and we assume that the magnetic induction vanishes outside $Q$.

Using trace and extension techniques from \cite{jonsson91} we establish a generalized Stokes formula, see Theorem \ref{stokes}. It involves generalized tangential traces that can be expressed as a limit of tangential traces along the boundaries of ``polyhedral'' approximations. We establish a Friedrichs inequality, Theorem \ref{friedrichs}, and establish existence and uniqueness of weak solutions, Theorem \ref{exandunique}. For the numerical approximation, we restrict ourselves to the axial-symmetric case, which in turn brings us to solve the problem in the snowflake domain. We consider both the fractal and pre-fractal problems, which we denote with $(\bar P)$ and $(\bar P_n)$ respectively. We prove existence and uniqueness of weak solutions (Propositions \ref{propex} and \ref{propexn}) and regularity results (Proposition \ref{propreg1}). We show that, in a suitable sense, the pre-fractal solutions converge to the limit fractal one, see Theorem \ref{convergenza}. We consider the numerical approximation of the pre-fractal problem $(\bar P_n)$ by a FEM scheme. To obtain an optimal a priori error estimate, we rely on the regularity of the weak solution of problem $(\bar P_n)$ in suitable weighted Sobolev spaces, see Theorem \ref{pesata}. Since the pre-fractal domain $\Omega_n$ is not convex, the solution is not in $H^2(\Omega_n)$, hence the rate of convergence is deteriorated. By using a suitable mesh constructed in \cite{cefalolancia}, which is compliant with the so-called \emph{Grisvard conditions} \cite{grisvard}, we can prove optimal a priori error estimates. These conditions involve the weight exponent of the weak solution given in Theorem \ref{pesata}. We finally present numerical simulations, which describe the behavior of the magnetic field. It turns out that the intensity of the magnetic field
increases as the length of the boundary approaches the $\lq\lq$length" of $F\times I$. We believe that this effect may be useful for potential applications.

\section{Fractal domains}\label{geometria}

We write $|P-P'|$ to denote the Euclidean distance between two points $P$ and $P'$ in $\R^N$. The \emph{Koch snowflake} $F\subset \R^2$ is the union $F=\bigcup_{i=1}^3 K^{(i)}$ of three com-planar Koch curves $K^{(1)}$, $K^{(2)}$ and $K^{(3)}$, cf. \cite[Chapter 8]{Fa}, whose junction points
  $A$, $B$ and $C$ are the vertices of a regular triangle. We assume this triangle has unit side length, i.e.
  $|A-B|=|A-C|=|B-C|=1$. 
  
The single \emph{Koch curve} $K^{(1)}$ is the uniquely determined self-similar set with respect to a family $\Psi^{1}$ of four contractive similarities $\psi_{1}^{(1)},...,\psi_{4}^{(1)}$, all having contraction ratio $\frac{1}{3}$, see \cite{Fa, FrLa}. Let $V_0^{(1)}:=\lbrace A, B \rbrace$, $\psi_{i_1\dots i_n}:=\psi_{i_1}\circ\dots\circ\psi_{i_n}$, $V_{i_1\dots i_n}^{(1)}:=\psi_{i_1\dots i_n}^{(1)}(V_0^{(1)})$ and
	\begin{center}
	$V_n^{(1)}:=\bigcup\limits_{i_1\dots i_n=1}^4 V_{i_1\dots i_n}^{(1)}$.
	\end{center}
We write $i|n=(i_1,i_2,\dots,i_n)$ and $V_{\star}^{(1)}:=\cup_{n\geq 0}V_n^{(1)}$. The closure in $\R^N$ of $V_\ast^{(1)}$ is just $K^{(1)}$. Now let $K_0^{(1)}$ denote the unit segment whose endpoints are $A$ and $B$. We set $K_{i_1\dots i_n}^{(1)}=\psi_{i_1\dots i_n}(K_0^{(1)})$ and 
\[K_n^{(1)}:=\bigcup\limits_{i_1\dots i_n=1}^4 K_{i_1\dots i_n}^{(1)}.\]

In a similar way, it is possible to approximate $K^{(2)}$ and $K^{(3)}$ by the sequences $(V_n^{(2)})_{n\geq 0}$ and $(V_n^{(3)})_{n\geq 0}$, we denote their unions by $V_{\star}^{(2)}$ and $V_{\star}^{(3)}$, respectively. The polygonal curves associated with $V_n^{(2)}$ and $V_n^{(3)}$ are denoted by $K_n^{(2)}$ and $K_n^{(3)}$, respectively.

The Koch snowflake $F$ itself is approximated by the sequence $(F_n)_{n\geq 1}$ of ``pre-fractal'' closed polygonal curves $F_n$, defined by 
\begin{equation}\label{eq:3bitris}
F_{n}=\displaystyle\bigcup_{i=1}^3K_n^{(i)},
\end{equation}
see Figure \ref{figura1}.

%

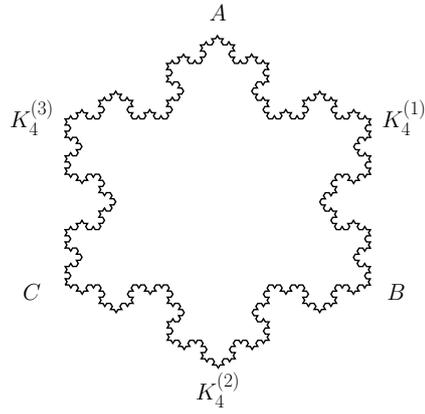
\begin{figure}
\centering
\scalebox{.35}{
\begin{tikzpicture}
\node  [anchor=south west] (label) at (0,0) {\includegraphics[width=350pt,height=380pt]{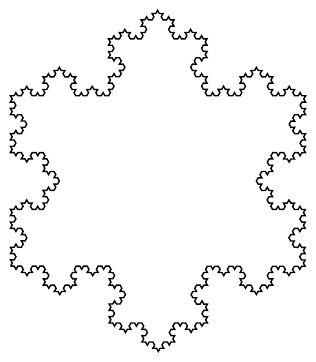}};
\node at (6.3,14) {{\Huge $A$} };
\node at (13,3.4) {{\Huge $B$} };
\node at (-0.7,3.4) {{\Huge $C$} };
\node at (6.3, -0.3) {{\Huge $K^{(2)}_4$} };
\node at (13.3, 10) {{\Huge $K^{(1)}_4$} };
\node at (-0.7, 10) {{\Huge $K^{(3)}_4$} };
\end{tikzpicture}
  }
  \caption{The pre-fractal curve $F_4$.}
  \label{figura1}
\end{figure}


\begin{notation}\label{notation-key}
	By $\Omega_n\subset\R^2$ we denote the bounded open set with boundary $F_n$ and by $Q_n$ the three-dimensional cylindrical domain having $S_n:=F_n\times [0,1]$ as ``lateral surface'' and the sets $\Omega_n\times\{0\}$ and $\Omega_n\times\{1\}$ as bases. We similarly write $\Omega$ for the bounded open domain in $\R^2$ with boundary $F$ (``snowflake domain''), define the cylindrical-type surface $S:=F\times I$ and let $Q$ denote the open cylindrical domain having $S$ as lateral surface and the sets $\Omega\times\{0\}$ and $\Omega\times\{1\}$ as bases, see Figure \ref{fig2}. 
\end{notation}


\begin{figure}
\centering
\includegraphics[width=0.50\textwidth]{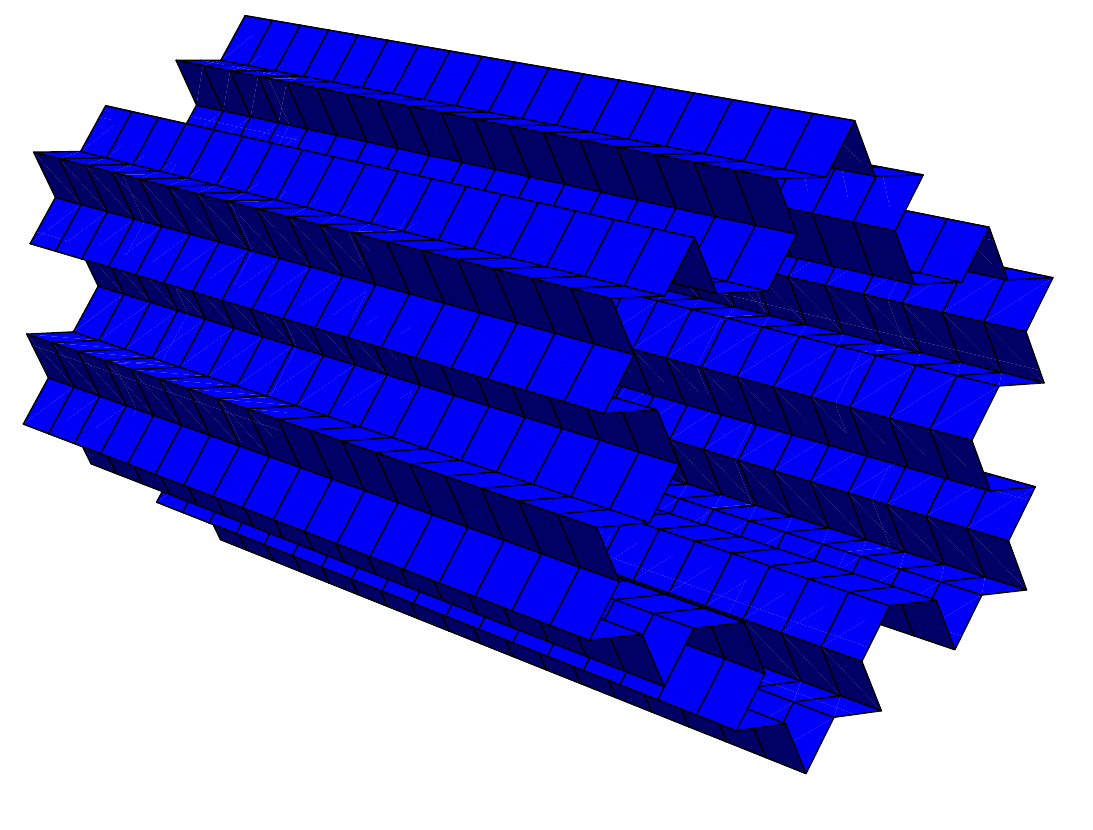}
\caption{The lateral surface $S_2$.}
\label{fig2}
\end{figure}

\section{A 3D magnetostatics problem}\label{modello3d}
\setcounter{equation}{0}
%

We formulate a linear magnetostatic problem on the fractal domain $Q$. To deduce it and to explain its physical meaning
we start by recalling \emph{Maxwell's equations} for classical macroscopic electromagnetic fields. We assume that $Q$ is made up from a \emph{linear} material, i.e. in a material without any magnetization or polarization effects, and we assume it is \emph{dielectric}, i.e. its conductivity can be neglected, see for instance \cite[Section 1.2.1]{Monk03}. Then \emph{Amp\`{e}re's law}, $\curl(\mathcal{H})=\mathcal{J}+\frac{\partial\mathcal{D}}{\partial t}$, tells that the total \emph{magnetic field} $\mathcal{H}$ induced around a closed loop equals the electric current plus the rate of change of the \emph{electric displacement field}  $\mathcal{D}$ enclosed by the loop, here $\mathcal{J}$ denotes the \emph{electric current density}, i.e. the vector field describing the directed flow of electric charges. The corresponding \emph{magnetic induction} is $\mathcal{B}=\mu\mathcal{H}$, where $\mu$ is a positive and bounded scalar function of space and time, called the \emph{permeability} of the material. By \emph{Faraday's law of induction}, $\curl(\mathcal{E})=-\frac{\partial\mathcal{B}}{\partial t}$, the voltage induced in a closed loop equals the change of the enclosed magnetic field. Here $\mathcal{E}=\frac{1}{\varepsilon}\mathcal{D}$, where $\varepsilon$ is a positive and bounded scalar of space and time referred to as the \emph{permittivity} of the material. These assumptions of $\mu$ and $\varepsilon$ mean we model a inhomogeneous isotropic material, so practically $Q$ may consist of a mixture of different materials whose electromagnetic properties may depend on the location in space but not on the direction of the fields. \emph{Gauss' law}, $\dive(\mathcal{D})=\rho$, states that the electric flux leaving a volume equals the charge inside, here $\rho\geq 0$ is the \emph{charge density}. According to \emph{Gauss' law for magnetism}, $\dive(\mathcal{B})=0$, i.e. the magnetic flux through a closed surface is zero.

We now make the following assumptions leading to a much simpler \emph{magnetostatic} setup:
\begin{itemize}
    \item the permittivity $\varepsilon=\varepsilon(x)$ and the permeability $\mu=\mu(x)$ are \emph{time-independent};
    \item the charge density is zero, $\rho=0$;
    \item the current density $\mathcal{J}\equiv {\bf J}(x)$ is \emph{time-independent} and \emph{real valued};
    \item the fields $\mathcal{E}\equiv {\bf E}(x)$ and $\mathcal{H}\equiv {\bf H}(x)$ are \emph{time-independent} and 
    \emph{real valued};
	\item all the fields vanish outside $Q$.	
\end{itemize}
Under these assumptions, Maxwell's equations on $Q$ read
\begin{equation}\label{maxwell2}
\curl(\bf{H})=\bf{J},\quad \curl({\bf E})=0,\quad \dive(\bf{D})=0,\quad \dive(\bf{B})=0,
\end{equation}
where ${\bf D}=\varepsilon{\bf E}$ and ${\bf B}=\mu{\bf H}$. 

Our assumption that ${\bf E}$ vanishes in $Q^c$ means that the surrounding region $Q^c$ is a \emph{perfect conductor}. When passing from one to another medium the parallel component of the electric field should be continuous, this can be seen by taking a small rectangular loop with long sides parallel to $\partial Q$, one inside $Q$, one outside and applying Faraday's law. Since the field vanishes outside $Q$ this forces to impose what is referred to as the \emph{perfectly conducting boundary condition} ${\bf n}\times{\bf E}=0$ on $\partial Q$. 

Since ${\bf B}$ is divergence free, there exists a \emph{magnetic vector potential} ${\bf u}=(u_1,u_2,u_3)$ such that
${\bf B}=\curl({\bf u})$, and we may choose it to be divergence free, $\dive {\bf u}=0$. Note that Gauss' law for magnetism then becomes trivial.

Also ${\bf B}$ is supposed to be zero on $Q^c$. Therefore, looking at the flux of the magnetic field through small closed loops on $\partial Q$, which should not differ for the interior and the exterior field, and applying the Kelvin-Stokes theorem, it follows that we should impose ${\bf n}\times{\bf u}=0$ on $\partial Q$. See for instance \cite[Section 5.4.2]{Griffiths} or \cite[p. 82]{vago}.

We now restrict attention to the magnetic field only and pose the following problem: Given $\mu$ and ${\bf J}$ as above, find a magnetic vector potential ${\bf u}$ that satisfies 
\begin{equation}\label{sistema 3d}
(P)\begin{cases}
\curl\left(\frac{1}{\mu}\curl({\bf u})\right)={\bf J}\quad &\text{in}\,\, Q,\\[2mm]
\dive{\bf u}=0\quad &\text{in}\,\, Q,\\[2mm]
{\bf n}\times{\bf u}=0 &\text{on}\,\,\partial Q.
\end{cases}
\end{equation}
Note that if $\mu$ is constant then the first equation rewrites 
\begin{equation}\label{E:vecLap}
-\Delta_{\vek}{\bf u}=\mu{\bf J},
\end{equation}
 where $\Delta_{\vek}$ denotes the vector Laplacian.

\section{Trace theorems, Stokes formula and Gauss-Green identity}
\setcounter{equation}{0}

We discuss measures, function spaces and trace theorems. The latter allow rigorous definitions of boundary conditions and generalizations of classical integral formulas. We write $B(P,r)=\{P'\in \R^N\,:\,|P'-P|<r \}$, $P \in \R^N, r>0$, for the euclidean ball of radius $r$ centered at $P$. For the two-dimensional Lebesgue measure we write $\de x_1\de x_2$ and for the three-dimensional one we write $\de x=\de x_1\de x_2\de x_3$. 

On the snowflake curve $\displaystyle F=\bigcup_{i=1}^3 K^{(i)}$ we consider the finite Borel measure $\mu$ defined by
\[\mu_F:=\mu_1+\mu_2+\mu_3,\]
where $\mu_i$ denotes the normalized Hausdorff measure of dimension $D_f=\frac{\ln 4}{\ln 3}$, restricted to $K_i$, $i=1,2,3$. It is well known that $c_1r^{D_f}\leq \mu_F(B(P,r))\leq c_2r^{D_f}$, $P\in F$, $r>0$, with positive constants $c_1$ and $c_2$. If we endow the cylindrical type surface $S=F\times I$ with the measure
\[\de \mu_S:=\de\mu_{F}\times\de x_3,\]
where $\de x_3$ is one-dimensional Lebesgue measure on $I$, then clearly 
\begin{equation}\label{E:muF}
c_1r^{D_f+1}\leq \mu_S(B(P,r))\leq c_2r^{D_f+1}
\end{equation}
for all $P\in S$ and $r>0$.

We equip the boundary $\partial Q$ with the measure 
 \begin{equation}\label{e-mu}
\de\mu_{\partial Q}=\chi_S\de \mu_S+\chi_{\tilde\Omega}\de x_1\de x_2,
\end{equation} 
where $\tilde\Omega=(\Omega\times\{0\})\cup(\Omega\times\{1\})$ is the union of the two bases of the cylinder domain $Q$  in Notation~\ref{notation-key}. In particular, $\supp \mu_{\partial Q}=\partial Q$. 

From \eqref{E:muF} and the quadratic scaling of the two-dimensional Lebesgue measure it follows that
\begin{equation}\label{E:scaling}
\mu_{\partial Q}(B(P,kr))\leq c_1\:k^{D_f+1}\mu_{\partial Q}(B(P,r)) \quad \text{and}\quad \mu_{\partial Q}(B(P,kr))\geq c_2\:k^2\mu_{\partial Q}(B(P,r))
\end{equation} 
for all $P\in \partial Q$, $r>0$, $k\geq 1$ such that $kr\leq 1$.

We write $L^2(Q)$ and $L^2(Q_n)$ for the $L^2$-spaces with respect to the three-dimensional Lebesgue measure, the spaces
$L^2(\Omega)$, $L^2(\Omega_n)$, $L^2(\partial Q_n)$ are taken with respect to the two-dimensional Lebesgue respectively Hausdorff measure (depending on whether considered in $\R^2$ or $\R^3$). For $\partial Q$ we write $L^2(\partial Q)=L^2(\partial Q,\mu_{\partial Q})$, the $L^2$-space with respect to $\mu_{\partial Q}$. 

The spaces $H^\alpha(\R^N)=H^{\alpha,2}(\R^N)$ denote the usual Bessel potential spaces, see for instance \cite{AdHei}, where they are denoted by $L^{\alpha,2}(\R^N)$. Given a domain $O\subset \R^N$, the notation $H^1(O)$ denotes the classical Sobolev space of square integrable functions with finite Dirichlet integral, usually denoted by $W^{1,2}(O)$.

{Since the boundary $\partial Q=\tilde\Omega\cup S$ is a closed set composed by sets of different Hausdorff dimension, in order to consider the trace space of $H^\alpha(Q)$ on $\partial Q$, we introduce suitable spaces $\tilde B_\alpha^{2,2}(\partial Q)$ as in \cite[page 356]{jonsson91}.} For any
\begin{equation}\label{E:smoothness}
\frac{1}{2}<\alpha<2-\frac{D_f}{2}
\end{equation}
let $\tilde B^{2,2}_\alpha(\partial Q)$ denote the class of functions $u$ on $\partial Q$ such that 
\begin{equation}\label{besov chiusi}
\|u\|^2_{\tilde B^{2,2}_\alpha(\partial Q)}=\|u\|^2_{L^2(\partial Q)}+\iint_{|x-y|<1}\frac{|u(x)-u(y)|^2}{|x-y|^{2\alpha-3}(\mu_{\partial Q}(B(x,|x-y|)))^2}\,\de\mu_{\partial Q}(x)\,\de\mu_{\partial Q}(y)
\end{equation}
is finite. 

{We remark that  $\mu_{\partial Q}$ defined 
in \eqref{e-mu}
	is not an Ahlfors regular  $d$-measure on $\partial Q$. That is,     the $\mu_{\partial\Omega}$-measure of a ball of radius $r>0$ can not be estimated from above and below, respectively, by a constant times $r^d$. 
	Therefore the space $\tilde B^{2,2}_\alpha(\partial Q)$ does not coincide with the usual Besov space $B^{2,2}_\alpha(\partial Q)$ defined in \cite[page 103]{JoWa} or \cite{triebel}}.

We denote by $|A|$ the Lebesgue measure of a subset $A\subset\R^N$. For  $f\in H^{\alpha}(O)$, $O\subset \R^N$ open, we put
\begin{equation} \label{e3.33}
\gamma_0f(P)=\lim_{r\to 0}{1\over|B(P,r)\cap O|}\int_{B(P,r)\cap O}f(x)\,\de x
\end{equation}
at every point $P\in \overline{O}$ where the limit exists. This is a typical form of \emph{restriction operator} in the spirit of Lebesgue differentiation.

The following trace theorem is a special case of \cite[Theorem 1]{jonsson91}, see also \cite[Proposition 2]{jonsson91}.

\begin{prop}
Let $\alpha$ be as in \eqref{E:smoothness}. $\tilde B_\alpha^{2,2}(\partial Q)$ is the trace space of $H^{\alpha}(\R^3)$, that is:
\begin{enumerate}
\item[(i)] $f\mapsto \gamma_0f$ is a linear and continuous operator from $H^\alpha(\R^3)$ to $\tilde B_{\alpha}^{2,2}(\partial Q)$;
\item[(ii)] There exists a linear and continuous operator $\Ext\colon\tilde B_{\alpha}^{2,2}(\partial Q)\to H^{\alpha}(\R^3)$ such that $\gamma_0 \circ\Ext$ is the identity operator on $\tilde B_{\alpha}^{2,2}(\partial Q)$. 
\end{enumerate}
\end{prop}

Combined with trace and extension results between the spaces $H^1(\R^3)$ and $H^1(Q)$, such as for instance \cite[Chapter VII, Theorem 1, combined with Chapter VIII, Proposition 1]{JoWa}, we obtain the following. 

\begin{folg}\label{P:trace}
The space $\tilde B_1^{2,2}(\partial Q)$ is the trace space of $H^1(Q)$ on $\partial Q$, i.e. there exist a continuous linear restriction operator from $H^1(Q)$ to $\tilde B_1^{2,2}(\partial Q)$ and a continuous linear extension from $\tilde B_1^{2,2}(\partial Q)$ to $H^1(Q)$.
\end{folg}

For the restriction to $\partial Q$ of a function $f\in H^1(Q)$ we write $f|_{\partial Q}$. 

More classical trace and extension results cover the case of Lipschitz boundaries, such as the sets $\partial Q_n:=S_n\cup \tilde\Omega_n$, where $\tilde{\Omega}_n:=(\Omega_n\times \left\lbrace 0\right\rbrace)\cup(\Omega_n\times \left\lbrace 1\right\rbrace)$. For the following result see \cite{grisvard-fr,necas}.


\begin{prop}  
The space $H^{\frac{1}{2}}(\partial Q_n)$ is the trace space of $H^1(Q_n)$ on $\partial Q_n$ in the
following sense:

\begin{enumerate}
\item[(i)] $\gamma_0$ is a continuous and linear operator
from  $H^{1}(Q_n)$ to $H^{\frac{1}{2}}(\partial Q_n)$; \item[(ii)]
there is a continuous linear operator $\Ext$ from
$H^{\frac{1}{2}}(\partial Q_n)$ to $H^{1}(Q_n)$ such that
$\gamma_0\circ\Ext$ is the identity operator in
$H^{\frac{1}{2}}(\partial Q_n)$.
\end{enumerate}
\end{prop}
 
As usual, we write $H^{-\frac{1}{2}}(\partial Q_n)$ to denote the dual space of $H^{\frac{1}{2}}(\partial Q_n)$, see \cite[p. 8]{GirRav86}.

We pass to vector valued functions. Consider the space
\[H(\curl,Q):=\left\lbrace {\bf u}=(u_1, u_2, u_3)\colon Q\to \mathbb{R}^3: u_1, u_2, u_3\in L^2(Q) \,\,\text{and}\,\curl{\bf u}\in L^2(Q)^3  \right\rbrace.\]
Endowed with the norm $\|{\bf u}\|_{\curl,Q}=\left(\left\|{\bf u}\right\|_{L^2(Q)^3}^2+\left\|\curl {\bf u}\right\|^2_{L^2(Q)^3}\right)^{1/2}$, it becomes a Hilbert space, see for instance \cite{duv-lions, GirRav86} or \cite{temam}.

We now prove a generalized vector Stokes formula. Suppose ${\bf{u}}\in H(\curl,Q)$. For any ${\bf v}\in \tilde B^{2,2}_{1}(\partial Q)^3$ let ${\bf w}\in H^1(Q)^3$ be such that ${\bf w}|_{\partial Q}={\bf v}$, defined component-wise in the sense of Corollary \ref{P:trace}, and consider the quantity
\[\gamma_\tau{\bf u}({\bf{v}}):=\int_Q{\bf{u}} \cdot \curl\,{\bf{w}}\,\de x -\int_Q {\bf{w}}\cdot \curl\, {\bf{u}} \,\de x.\]


\begin{theorem}\label{stokes} Let $Q$ be the Koch-type pipe.
\begin{enumerate}
\item[(i)] The map ${\bf u} \mapsto \gamma_\tau\bf{u}$ is well defined as a bounded linear operator from $H(\curl,Q)$ into $(( \tilde B^{2,2}_{1}(\partial Q))')^3$. By setting ${\bf{u}}\times{\bf{n}}|_{\partial Q}:= \gamma_\tau\bf{u}$
we have
\begin{equation}\label{stima}
|\left\langle {\bf{u}\times\bf{n}}|_{\partial Q},{\bf{v}}\right\rangle_{(( \tilde B^{2,2}_1(\partial Q))')^3,  \tilde B^{2,2}_{1}(\partial Q)^3}|\leq c\left\|{\bf{u}}\right\|_{\curl, Q}\left\|{\bf{v}}\right\|_{\tilde B^{2,2}_{1}( \partial Q)^3}
\end{equation}
for all ${\bf u}\in H(\curl,Q)$ and ${\bf v}\in  \tilde B^{2,2}_{1}(\partial Q)^3$. 
\item[(ii)] Moreover, we have 
\begin{equation}\label{E:approx}
\left\langle {\bf{u}}\times{\bf{n}}|_{\partial Q},{\bf w}|_{\partial Q}\right\rangle_{(( \tilde B^{2,2}_1(\partial Q))')^3, \tilde B^{2,2}_{1}(\partial Q)^3}=\lim_{n\to\infty}\left\langle {\bf u}\times {\bf n}|_{\partial Q_n}, {\bf{w}}|_{\partial Q_n}\right\rangle_{H^{-\frac{1}{2}}(\partial Q_n)^3, H^\frac{1}{2}(\partial Q_n)^3}
\end{equation}
and 
\begin{equation}\label{E:Stokes2}
\left\langle {\bf{u}}\times{\bf{n}}|_{\partial Q},{\bf w}|_{\partial Q}\right\rangle_{(( \tilde B^{2,2}_1(\partial Q))')^3, \tilde B^{2,2}_{1}(\partial Q)^3}=\int_{Q} {\bf u}\cdot\curl {\bf{w}}\,\de x-\int_{Q} {\bf{w}}\cdot\curl {\bf u}\,\de x
\end{equation}
for all ${\bf u}\in H(\curl,Q)$ and ${\bf{w}}\in H^1(Q)^3$.
\end{enumerate}
\end{theorem}
\medskip

Formula \eqref{E:approx} provides a suitable  approximation of ${\bf{u}}\times{\bf{n}}|_{\partial Q}$ in terms of the tangential traces ${\bf u}\times {\bf n}|_{\partial Q_n}$ along the Lipschitz boundaries $\partial Q_n$, see \cite[\S 2, Theorem 2.11]{GirRav86} or \cite{temam}. In this sense ${\bf u}\times {\bf n}|_{\partial Q}$ can be seen as a \emph{generalized tangential trace} and \eqref{E:Stokes2} is a \emph{generalized Stokes formula}. 



\begin{proof}
Let ${\bf u}\in H(\curl, Q)$. Given ${\bf v}\in \tilde B^{2,2}_{1}(\partial Q)^3$, let ${\bf{w}}\in H^1(Q)^3$ be such that ${\bf w}|_{\partial Q}={\bf v}$ in $\tilde B^{2,2}_{1}(\partial Q)^3$. Then Cauchy-Schwarz together with the inclusion $H^1(Q)^3\subset H(\curl,Q)$ and {Corollary} \ref{P:trace} lead to the estimate
\begin{align}
|\left\langle {\bf{u}}\times{\bf{n}}|_{\partial Q},{\bf w}|_{\partial Q}\right\rangle|&\leq \|{\bf u}\|_{L^2(Q)^3}\|\curl {\bf{w}}\|_{L^2(Q)^3}+\|{\bf{w}}\|_{L^2(Q)^3}\|{\bf \curl u}\|_{L^2(Q)^3}\notag\\
&\leq c\,\|{\bf w}\|_{H^1(Q)^3}\|{\bf u}\|_{\curl,Q}\notag\\
&\leq c\,\|{\bf v}\|_{ \tilde B^{2,2}_1 (\partial Q)^3}\|{\bf u}\|_{\curl,Q}.\notag
\end{align}
This shows in particular, that $\gamma_\tau{\bf u}({\bf{v}})$ is independent from the choice of the extension ${\bf w}$ of ${\bf v}$, and that $\bf{u}\times\bf{n}$ is an element of $(( \tilde B^{2,2}_{1}(\partial Q))')^3$ which satisfies \eqref{stima}.

We now consider the sequence of domains $Q_n=\Omega_n\times I$, which are bounded Lipschitz domains and satisfy $Q_n\subset Q_{n+1}$ and $Q=\bigcup_{n=1}^\infty Q_n$.
By the vector Stokes formula for Lipschitz domains, cf.  \cite[\S 2, Theorem 2.11]{GirRav86} or Appendix I in \cite{temam}, together with the dominated convergence theorem, we have 
\begin{align}\label{E:Lipschitz}
\lim_{n\to\infty} \left\langle {\bf u}\times {\bf n}|_{\partial Q_n}, {\bf{w}}|_{\partial Q_n}\right\rangle_{H^{-\frac{1}{2}}(\partial Q_n)^3, H^\frac{1}{2}(\partial Q_n)^3}&=\lim_{n\to\infty}\int_{Q_n} {\bf u}\cdot\curl {\bf{w}}\,\de x-\int_{Q_n} {\bf{w}}\cdot\curl {\bf u}\,\de x\notag\\
&=\int_{Q} {\bf u}\cdot\curl {\bf{w}}\,\de x-\int_{Q} {\bf{w}}\cdot\curl {\bf u}\,\de x\notag\\
&=\left\langle {\bf{u}\times\bf{n}},{\bf w}|_{\partial Q}\right\rangle_{(( \tilde B^{2,2}_1(\partial Q))')^3, \tilde B^{2,2}_{1}(\partial Q)^3}\notag
\end{align}
for all ${\bf{w}}\in H^1(Q)^3$ and $n$, where ${\bf u}\times {\bf n}|_{\partial Q_n}$ is defined as an element of $H^{-\frac{1}{2}}(\partial Q_n)^3$.
\end{proof}

Next, consider the space 
\[H(\dive, Q):=\left\lbrace {\bf u}=(u_1,u_2,u_3) \colon Q\to \mathbb{R}^3: u_1, u_2, u_3\in L^2(Q) \,\,\text{and}\,\dive{\bf u}\in L^2(Q)\right\rbrace,\]  
which is Hilbert when equipped with the norm $\|{\bf u}\|_{\dive,Q}=\left(\left\|{\bf u}\right\|_{L^2(Q)^3}^2+\left\|\dive {\bf u}\right\|^2_{L^2(Q)}\right)^{1/2}$. Following the same pattern as above one can establish a generalized Gauss-Green formula. This can be done as in \cite{LaVe2}.

Suppose ${\bf u}\in H(\dive, Q)$. For any $v\in \tilde B^{2,2}_1(\partial Q)$ let $w\in H^1(Q)$ be such that $w|_{\partial Q}=v$ in the sense of {Corollary} \ref{P:trace} and consider
\[\gamma_\nu{\bf u}(v):=\int_Q {\bf u}\cdot \nabla w\:\de x + \int_Q(\dive {\bf u}) w\:\de x.\]
By proceeding as in \cite[{Theorem 3.7}]{LaVe2} we can prove the following Green formula.
\begin{theorem}\label{gauss-green} Let $Q$ be the Koch-type pipe.
\begin{enumerate}
\item[(i)] The map ${\bf u} \mapsto \gamma_\nu\bf{u}$ is well defined as a bounded linear operator from $H(\dive,Q)$ into $(( \tilde B^{2,2}_{1}(\partial Q))')$. By setting ${\bf{u}}\cdot{\bf{n}}|_{\partial Q}:=\gamma_\nu{\bf u}$, we have 
\begin{equation}
|\left\langle {\bf{u}}\cdot{\bf{n}}|_{\partial Q},v\right\rangle_{(( \tilde B^{2,2}_1(\partial Q))'), \tilde B^{2,2}_{1}(\partial Q)}|\leq c\left\|{\bf{u}}\right\|_{\dive, Q}\left\|v\right\|_{ \tilde B^{2,2}_{1}( \partial Q)}. \notag
\end{equation}
for all ${\bf u}\in H(\dive,Q)$ and $v\in \tilde B^{2,2}_{1}(\partial Q)$. 
\item[(ii)] Moreover, we have 
\begin{equation}\label{E:approx2}
\left\langle {\bf{u}}\cdot{\bf{n}}|_{\partial Q},w|_{\partial Q}\right\rangle_{(( \tilde B^{2,2}_1(\partial Q))'), \tilde B^{2,2}_{1}(\partial Q)}=\lim_{n\to\infty} \left\langle {\bf u}\cdot {\bf n}|_{\partial Q_n}, {w}|_{\partial Q_n}\right\rangle_{H^{-\frac{1}{2}}(\partial Q_n), H^\frac{1}{2}(\partial Q_n)}
\end{equation}
and 
\begin{equation}\label{E:GaussGreen}
\left\langle {\bf{u}}\cdot{\bf{n}}|_{\partial Q},w|_{\partial Q}\right\rangle_{(\tilde B^{2,2}_1(\partial Q))', \tilde B^{2,2}_{1}(\partial Q)}=\int_{Q}  {\bf u}\cdot\nabla w\,\de x-\int_{Q} (\dive {\bf u})w\,\de x
\end{equation}
for all ${\bf u}\in H(\dive,Q)$ and $w\in H^1(Q)$.
\end{enumerate}
\end{theorem}
\medskip

Similarly as before formula \eqref{E:approx2} provides a suitable approximation of ${\bf{u}}\cdot{\bf{n}}|_{\partial Q}$ by normal traces ${\bf{u}}\cdot{\bf{n}}|_{\partial Q_n}$ on the Lipschitz boundaries $\partial Q_n$, which follows again from corresponding results in the Lipschitz case, \cite[\S 2, Theorem 2.5]{GirRav86}.

\begin{remark} {We point out that the results of this Section hold not only for the Koch-type pipe. Indeed, these results can be extended to every domain $Q$ having as boundary $\partial Q$ a $d$-set or an arbitrary closed set of $\R^3$, under the assumption that $Q$ can be approximated by an invading sequence of Lipschitz domains $\{Q_n\}$, as in this case.}
\end{remark}

\section{Friedrichs inequality and weak solutions}
\setcounter{equation}{0}


We discuss \eqref{sistema 3d} in terms of weak solutions and the Lax-Milgram Theorem, and to do so we introduce the symmetric bilinear form
\begin{equation}
a({\bf u},{\bf w}) =\int_Q \curl({\bf w})\cdot\left(\frac{1}{\mu}\,\curl({\bf u})\right)\,\de x, \quad {\bf u}, {\bf w}\in H(\curl, Q),\notag
\end{equation}
where, in agreement with the above assumptions, $\mu$ is a real valued measurable function on $Q$ satisfying $\mu_0\leq\mu\leq\mu_1$ a.e. in $Q$ with two constants $\mu_0,\mu_1>0$. Given ${\bf J}\in L^2(Q)^3$ we consider the linear and continuous functional on $H(\curl, Q)$, defined by 
\[f({\bf w})=\int_Q {\bf J} \cdot {\bf w}\,\de x, \quad {\bf w}\in H(\curl, Q).\]

The interpretation as an identity in 
$(( \tilde B^{2,2}_{1}(\partial Q))')^3$ gives a rigorous meaning to the boundary condition ${\bf u}\times {\bf n}=0$ in \eqref{sistema 3d}. To encode it in a suitable function space, we consider the space $H_0(\curl, Q)$, defined as the closure in $H(\curl, Q)$ of all compactly supported smooth vector fields $C_c^\infty(Q)^3$.

\begin{remark}\label{R:curltrace}
Taking into account the boundary condition in \eqref{sistema 3d}, the natural space would be $\Ker \gamma_\tau:=\{{\bf w}\in H(\curl, Q)\,:\,{\bf n} \times {\bf w}=0\,\,\text{on}\;\partial Q \}$. The inclusion $H_0(\curl, Q)\subset \Ker \gamma_\tau$ follows from (\ref{E:Stokes2}). The reverse inclusion is not straightforward, and to keep the present note 
simple we leave its investigation to a later forthcoming paper.
\end{remark}

If we agree to say that a weak solution in $H_0(\curl, Q)$ of the equation 
\begin{equation}\label{E:singleeq}
\curl\left(\frac{1}{\mu}\curl({\bf u})\right)={\bf J}
\end{equation}
is a vector field ${\bf u}\in H_0(\curl, Q)$ such that $a({\bf u},{\bf v})= f({\bf v})$ for all ${\bf v}\in  H_0(\curl, Q)$, then test vector fields ${\bf v}$ can in particular be recruited from
\[\Ker(\curl, Q):=\left\lbrace {\bf w}\in H_0(\curl, Q): \curl {\bf w}=0\right\rbrace, \]  
so that a weak solution of $(P)$ can only exist if $\mathbf{J}$ satisfies the \emph{compatibility condition}
\begin{equation}\label{E:compatibility}
f({\bf v})=\int_Q \mathbf{J}\cdot {\bf v}\,\de x=0 \quad\forall\, {\bf v} \in\Ker(\curl, Q).
\end{equation}
Moreover, since we are also interested in the uniqueness of weak solutions, we restrict ourselves to the quotient space $H_0(\curl, Q)/\Ker(\curl, Q)$, which by a simple quadratic variational problem, \cite[Corollary 1.2]{GirRav86}, involving the quotient space norm, see  \cite[p. 94-95]{hiptmair} or \cite[Lemma 3.5]{Lukas03}, is seen to be isometrically isomorphic to the space
\begin{equation}\label{E:H0bot}
H_{0, \bot}(\curl, Q):=\left\{ {\bf u} \in H_0(\curl, Q)\, :\, \int_Q {\bf u}\cdot {\bf w}\,\de x=0 \  \text{ for all ${\bf w}\in \Ker(\curl,Q)$}\right\}.
\end{equation}

A second requirement to be incorporated in the function spaces is that a solution ${\bf u}$ of $(P)$ should be divergence free. We consider the space $H_0(\dive, Q)$, defined as the completion in $H(\dive, Q)$ of $C_c^\infty(Q)^3$, and its subspace  
\[\Ker(\dive, Q):=\{{\bf u}\in H_0(\dive,Q): \dive {\bf u}=0\}.\] 
This discussion suggests that one possible way to phrase $(P)$ rigorously could be to look for a weak solution to equation (\ref{E:singleeq}) in the space $H_{0, \bot}(\curl, Q)\cap \Ker(\dive, Q)$. The latter space admits a much simpler description. A proof of the following fact can be found at the end of this section.

\begin{proposition}\label{P:coincidence}
A vector field ${\bf u}\in H_0(\curl,Q)\cap H_0(\dive, Q)$ is an element of $H_{0,\bot}(\curl, Q)$ if and only if $\dive {\bf u}=0$.
\end{proposition}

\medskip

As a next step of simplification, the intersection of the spaces $H_0(\curl, Q)$ and $H_0(\dive, Q)$ can be determined in a standard way, see \cite[Theorem 2.5]{amrouche} or \cite[Lemma 2.5]{GirRav86}. As a by-product we obtain the following \emph{Friedrichs inequality}, \cite{Schw16}, sometimes also referred to as a \emph{Maxwell inequality}, \cite{NPW15}, which provides a suitable coercivity bound for our problem. As usual, $H_0^1(Q)$ denotes the closure of $C_c^\infty(Q)$ in $H^1(Q)$. 

\begin{theorem}\label{friedrichs}
We have $H_0(\curl, Q)\cap H_0(\dive, Q)=H_0^1(Q)^3$, and there exists a constant $C>0$ such that for any ${\bf u}\in H_0^1(Q)^3$, we have
\begin{equation}\label{gaffney}
\displaystyle\|{\bf u}\|_{H^1(Q)}\leq C\, (\|\curl\,{\bf u}\|_{L^2(Q)^3}+ \|\dive\,{\bf u}\|_{L^2(Q)}).
\end{equation}
In particular, we have $\displaystyle\|{\bf u}\|_{\curl,Q}\leq C\, \|\curl\,{\bf u}\|_{L^2(Q)^3}$ for all ${\bf u}\in H_0^1(Q)^3\cap \Ker(\dive, Q)$.
\end{theorem}

\begin{proof} We follow the cited references to prove $H_0(\curl, Q)\cap H_0(\dive, Q)\subset H_0^1(Q)^3$, the other inclusion is trivial. Given ${\bf u}\in H_0(\curl, Q)\cap H_0(\dive, Q)$ consider the trivial extension of ${\bf u}$ to $\R^3$,
\begin{center}
${\bf\tilde{u}}=
\begin{cases}
{\bf u}\quad &\text{in}\,\,Q,\\[2mm]
0 &\text{in}\,\,\R^3\setminus\overline Q.
\end{cases}$
\end{center}
Since ${\bf u}\in H_0(\curl, Q)\cap H_0(\dive, Q)$, it evidently follows that $\curl{\bf\tilde{u}}\in L^2(\R^3)^3$ and $\dive{\bf\tilde{u}}\in L^2(\R^3)$. By definition ${\bf\tilde{u}}$ has compact support (in the distributional sense), so that by Schwartz' Paley-Wiener Theorem, \cite[Theorem 7.3.1]{Hoermander}, the Fourier transform ${\bf\hat{u}}$ of ${\bf\tilde{u}}$ is analytic. The above properties can be rewritten algebraically as 
\begin{center}
$(\xi_2\hat{u}_3-\xi_3\hat{u}_2,\xi_3\hat{u}_1-\xi_1\hat{u}_3,\xi_1\hat{u}_2-\xi_2\hat{u}_1)\in L^2(\R^3)^3$\quad and \quad$\xi_1\hat{u}_1+\xi_2\hat{u}_2+\xi_3\hat{u}_3\in L^2(\R^3)$.
\end{center}
It then follows that, for $i,j=1,2,3$,
\begin{equation}\label{E:target}
\|\xi_i\hat{u}_j\|_{L^2(\R^3)}\leq\|\curl{\bf\tilde{u}}\|_{L^2(\R^3)^3}+\|\dive{\bf\tilde{u}}\|_{L^2(\R^3)}.
\end{equation}
Note that for instance $(\xi_1\hat{u}_2-\xi_2\hat{u})^2\geq (\xi_1\hat{u}_2)^2-[(\xi_1\hat{u}_1)^2+(\xi_2\hat{u}_2)^2]+(\xi_2\hat{u}_1)^2$, and by rearranging and summing up we obtain \eqref{E:target}.
It follows that
\begin{equation*}
\|\nabla {\bf u}\|_{L^2(Q)^3}\leq\|\curl {\bf u}\|_{L^2(Q)^3}+\|\dive{{\bf u}}\|_{L^2(Q)}.
\end{equation*}
Hence ${\bf u}\in H^1_0(Q)^3$, and using Poincar\'e' inequality for $Q$ we obtain (\ref{gaffney}).
\end{proof}

We say that ${\bf u}$ is a \emph{weak solution of $(P)$} if ${\bf u}\in H_0^1(Q)^3\cap \Ker(\dive, Q)$ and $a({\bf u},{\bf v})=f({\bf v})$ for all ${\bf v}\in H_0^1(Q)^3\cap \Ker(\dive, Q)$. 

Existence and uniqueness of a solution are now easily seen from the Lax-Milgram Theorem  (see \cite{quarval}) together with Theorem \ref{friedrichs}. 

\begin{theorem}\label{exandunique}
For any ${\bf J}\in L^2(Q)^3$ satisfying (\ref{E:compatibility}) there exists a unique weak solution ${\bf u}$ of problem (P). Moreover, there exists a positive constant $C=C(Q,\mu_0,\mu_1)$ such that
\begin{center}
$\displaystyle\|{\bf u}\|_{\curl,Q}\leq C \|{\bf J}\|_{L^2(Q)^3}$.
\end{center}
\end{theorem}

The rest of this section is devoted to the proof of Proposition \ref{P:coincidence}. The first observation follows from  (\ref{E:GaussGreen}) by the same arguments as used to show \cite[Theorem 2.6]{GirRav86}, we recall them for convenience.
Let $\Ker \gamma_\nu:=\{{\bf w}\in H(\dive, Q)\,:\,{\bf n} \cdot {\bf w}=0\,\,\text{on}\;\partial Q \}$.
\begin{theorem}\label{T:A}
We have $H_0(\dive, Q)=\Ker \gamma_\nu$.
\end{theorem}
\begin{proof}
It suffices to show that $C_c^\infty(Q)^3$ is dense in $\Ker \gamma_\nu$. Let $l\in (\Ker \gamma_\nu)'$ and let ${\bf v}\in \Ker \gamma_\nu$ be such that 
\[\left\langle l,{\bf u}\right\rangle_{(\Ker \gamma_\nu)',\Ker \gamma_\nu}=\int_Q {\bf v}\cdot{\bf u}\,\de x+\int_Q \widetilde{v}\dive {\bf u}\,\de x,\quad {\bf u}\in \Ker \gamma_\nu,\]
where $\widetilde{v}=\dive {\bf v}$. Suppose now that $l\equiv 0$ on $C_c^\infty(Q)^3$. Then ${\bf v}=\nabla \widetilde{v}$ in distributional sense on $Q$, and since ${\bf v}\in L^2(Q)^3$, it follows that $\widetilde{v}\in H^1(Q)$. By (\ref{E:GaussGreen}) therefore 
\[\left\langle l,{\bf u}\right\rangle_{(\Ker \gamma_\nu)',\Ker \gamma_\nu}=\left\langle {\bf{u}}\cdot{\bf{n}}|_{\partial Q},\widetilde{v}|_{\partial Q}\right\rangle_{(\tilde B^{2,2}_1(\partial Q))', \tilde B^{2,2}_{1}(\partial Q)}=0,\quad {\bf u}\in \Ker \gamma_\nu.\]
This implies the desired density, see \cite[p. 26, property (2.14)]{GirRav86}. 
\end{proof}

The second item is an adaption of \cite[Theorem 2.7]{GirRav86} about the complement of $\Ker(\dive, Q)$, seen as a closed subspace of $L^2(Q)^3$. Again we briefly recall the classical proof.
\begin{theorem}\label{T:B}
The space $L^2(Q)^3$ admits the orthogonal decomposition 
\[L^2(Q)^3=\Ker(\dive,Q)\oplus\left\lbrace \nabla q: q\in H^1(Q)\right\rbrace.\]
\end{theorem}
\begin{proof}
The space $X:=\left\lbrace \nabla q: q\in H^1(Q)\right\rbrace$ is a closed subspace of $L^2(Q)^3$, so it suffices to show that $X^\bot=H:=\Ker(\dive,Q)$. If ${\bf u}\in H$, then by (\ref{E:GaussGreen}) and Theorem \ref{T:A} we have 
\begin{equation}\label{E:kill}
\int_Q {\bf u}\cdot \nabla q\, \de x=0,\quad q\in H^1(Q),
\end{equation}
so that $H\subset X^\bot$. If ${\bf u}\in L^2(Q)^3$ satisfies (\ref{E:kill}), then taking $q\in C_c^\infty(Q)^3$ implies $\dive {\bf u}=0$ and in particular, ${\bf u}\in H(\dive, Q)$, so that (\ref{E:GaussGreen}) may be applied and yields ${\bf u}\cdot{\bf n}=0$, i.e. ${\bf u}\in H_0(\dive, Q)$ and therefore ${\bf u}\in H$. This shows $X^\bot=H$.
\end{proof}

Adaptions of \cite[Theorem 2.9 and Corollary 2.9]{GirRav86} provide a suitable version of the classical fact that a curl free differentiable vector field in a simply connected domain is a gradient field. We interpret $\curl$ as an operator on $L^2(Q)^3$ in the sense of distributions on $Q$. 
\begin{theorem}\label{T:C}
A vector ${\bf u}\in L^2(Q)^3$ satisfies $\curl {\bf u}=0$ if and only if there exists a function $q\in H^1(Q)/\mathbb{R}$ such that ${\bf u}=\nabla q$. 
\end{theorem}
\begin{proof}
If ${\bf u}=\nabla q$ with some $q\in H^1(Q)$ then clearly $\curl {\bf u}=0$. 

Suppose ${\bf u}\in L^2(Q)^3$ is such that $\curl {\bf u}=0$. Let $\widetilde{{\bf u}}$ be the extension of ${\bf u}$ to $\R^3$ by zero on $Q^c$ and let $(\varrho_\varepsilon)_{\varepsilon>0}\subset C_c^\infty(\R^3)$ be a standard mollifier. Then we have $\curl\varrho_\varepsilon\ast{\widetilde{\bf u}}=\varrho_\varepsilon\ast\curl{\widetilde{\bf u}}$ and $\varrho_\varepsilon\ast{\widetilde{\bf u}}\in C_c^\infty(\R^3)^3$ for any $\varepsilon>0$, and $\lim_{\varepsilon\to 0} \varrho_\varepsilon\ast{\widetilde{\bf u}}=\widetilde{{\bf u}}$ in $L^2(Q)^3$.

Let $(O_n)_n$ be an increasing sequence of simply connected Lipschitz domains $O_n$ such that $\overline{O}_n\subset Q$ for all $n$ and $Q=\bigcup_{n=1}^\infty O_n$. Because the two-dimensional snowflake domain can be exhausted by increasing simply connected Lipschitz domains whose closures are contained in the snowflake domain, see for instance \cite[Section 6]{JMAA}, it follows easily that such a sequence $(O_n)_n$ exists.

If now $n$ is fixed and $\varepsilon>0$ is small enough then $\bigcup_{x\in O_n} B(x,\varepsilon)\subset Q$ and therefore $\curl  \varrho_\varepsilon\ast{\widetilde{\bf u}}=0$ in $O_n$. Consequently there is a function $q_\varepsilon\in H^1(O_n)$ such that $\varrho_\varepsilon\ast{\widetilde{\bf u}}=\nabla q_\varepsilon$ in $O_n$. Since $\lim_{\varepsilon\to 0} \nabla q_\varepsilon=\widetilde{{\bf u}} \in L^2(O_n)^3$, the limit $q_n:=\lim_{\varepsilon\to 0} q_{\varepsilon}$ exists in $H^1(O_n)/\R$, and clearly ${\bf u}=\nabla q_n$ in $O_n$. 

Varying $n$, we have $\nabla q_n=\nabla q_{n+1}$ in $O_n$, i.e. $q_n-q_{n+1}$ is constant on $O_n$. We can choose these constants so that $q_{n+1}=q_n$ in $O_n$ for all $n\geq 1$, and then consistently define $q:=q_n$ on $O_n$ for all $n\geq 1$ to obtain a function $q$ with the desired properties.
\end{proof}

Theorem \ref{T:C} implies a description of $\Ker(\curl, Q)$.
\begin{folg}\label{C:kercurl}
We have 
\[\Ker(\curl,Q)=\left\lbrace w\in H_0(\curl, Q): w=\nabla q\quad \text{for some $q\in H^1(Q)$}\right\rbrace.\]
\end{folg}

We can now easily prove Proposition \ref{P:coincidence}.
\begin{proof}
If ${\bf u}\in H_0(\curl, Q)\cap H_0(\dive, Q)$ is in $H_{0,\bot}(\curl, Q)$ then by (\ref{E:GaussGreen}) and Corollary \ref{C:kercurl} it satisfies 
\[\int_Q (\dive {\bf u})q\,\de x=\int_Q {\bf u}\cdot \nabla q\,\de x=0\]
for all $q\in H^1(Q)$ such that $\nabla q\in H_0(\curl, Q)$, and in particular, for all $q\in C_c^\infty(Q)$, what implies $\dive {\bf u}=0$ in $L^2(Q)$. The opposite inclusion follows similarly from (\ref{E:GaussGreen}).
\end{proof}

\begin{remark} Using \cite[Theorem 3]{wallin} one can show that $H_0^1(Q)$ coincides with the space of all elements of $H^1(Q)$ having zero trace on $\partial Q$. With Remark \ref{R:curltrace} and Theorem \ref{T:A} in mind one can therefore view Theorem \ref{friedrichs} as a rough paraphrase of the statement that if in the formal identity
${\bf u}|_{\partial Q}={\bf n}({\bf u}\cdot{\bf n})|_{\partial Q}+{\bf n}\times{\bf u}|_{\partial Q}$ both summands on the right hand side are zero, then we have ${\bf u}|_{\partial Q}=0$ in the sense of traces.
\end{remark}

\section{Weak solutions and H\"older regularity in 2D}
\setcounter{equation}{0}

We now reduce the three-dimensional problem $(P)$ to a magnetostatic problem in 2D. If ${\bf J}(x)=(0,0,J(x_1,x_2))$ and $\mu=\mu(x_1,x_2)$ then it is reasonable to assume that also the magnetic induction ${\bf B}$ does not depend on the $x_3$ coordinate. Therefore it is possible to choose a magnetic vector potential of form ${\bf u}=(0,0,u(x_1,x_2))$. Problem $(P)$ then reduces to finding a function $u=u(x_1,x_2)$ on $\Omega$ such that
\begin{equation}\label{sistema 2d}
(\bar P)\begin{cases}
-\dive \left(\frac{1}{\mu}\nabla u\right)=J\quad &\text{in}\,\, \Omega,\\[2mm]
u=0 &\text{on}\,\,\partial\Omega.
\end{cases}
\end{equation}
From this two-dimensional problem we obtain a magnetic induction of form ${\bf B}=(u_{x_2}, -u_{x_1},0)$. The domain $\Omega= \{(x_1,x_2)\in \R^2\,:\,(x_1,x_2,0)\in Q \}$ is a cross section of $Q$, i.e. $\Omega\times\{0\}= Q\cap \{x\in \R^3\,:\,x_3=0\}$, and the differential operator $\nabla u$ (applied to the scalar function $u$) operates only on the variables $x_1$ and $x_2$, i.e. $\nabla u=(u_{x_1}, u_{x_2})$. 

The energy form associated with $(\bar P)$ is 
\begin{equation}
a(u,v)=\int_\Omega \frac{1}{\mu(x)}\nabla u\nabla v\,\de x, \quad u,v\in H^1_0(\Omega),
\label{definforma2d}
\end{equation}
where as usual, $H^1_0(\Omega)$ denotes the closure in $H^1(\Omega)$ of the smooth functions with compact support in $\Omega$. 

\begin{prop}\label{propex}
For every given $J\in L^2(\Omega)$, there exists a unique weak solution in $H^1_0(\Omega)$ of problem $(\bar P)$, i.e. a function $u\in H^1_0(\Omega)$ such that
\begin{center}
$\displaystyle a(u,v)=\int_\Omega J\,v\,\de x, \quad v\in H^1_0(\Omega)$.
\end{center}
\end{prop}

We recall some regularity results for the weak solution of problem $(\bar P)$.

\begin{prop}\label{propreg1}
Suppose that $\mu$ is constant. Then the weak solution $u$ of problem $(\bar P)$ belongs to $W^{1,3}_0(\Omega)\cap C^{0,1/3}(\overline\Omega)$. Moreover $\nabla^2 u\in L^2(\Omega, d)$, where $d$ is the distance from the boundary. In particular it follows that ${\bf B}\in (L^3(Q))^3$.
\end{prop}

Here $W^{1,3}_0(\Omega)$ and $C^{0,1/3}(\overline\Omega)$ denote respectively the usual Sobolev space and the space of H\"older continuous functions of exponent $\frac{1}{3}$, while $\nabla^2 u$ denotes the Hessian of $u$. The statement $\nabla^2 u\in L^2(\Omega, d)$ means that 
\begin{equation*}
\int_\Omega |\nabla^2 u|^2 d(x,\partial\Omega)^2\,\de x<\infty.
\end{equation*}
For the proof of Proposition \ref{propreg1} we refer to Theorem 1.3 (part B) and Proposition 7.1 in \cite{Ny} (which is also related to  \cite{Ny96}).
{These references also explain the appearance of the exponents $1/3$  and $3$ in this proposition in relation to the geometry of the Koch snowflake. The proof of Nystrom's result is very technical and it is strictly
related to the Koch snowflake and to certain sophisticated estimates.  We
mentioned this result only for the sake of completeness, since we do not
use it for our results in the paper and do not need this type of a deeper analysis.}

We now consider the approximating problems on the pre-fractal domains $\Omega_n$ introduced in Section \ref{geometria}.

Let us assume that $\mu$ is a positive constant and $J\in L^2(\Omega)$. For every fixed $n\in \N$, we consider the following problems $(\bar P_n)$:

\begin{equation}\label{sistema 2d-h}
(\bar P_n)\begin{cases}
-\dive \left(\frac{1}{\mu}\nabla u_n\right)=J\quad &\text{in}\,\, \Omega_n,\\[2mm]
u_n=0 &\text{on}\,\,\partial\Omega_n.
\end{cases}
\end{equation}

We set $H^1_0(\Omega_n):=\overline{ \{ w\in C^1_0(\Omega): \supp w\subset \Omega_n \}}^{H^1(\Omega)}$. For every $u_n,v\in H^1_0(\Omega_n)$, let $$a_n(u_n,v)=\int_{\Omega_n}\frac{1}{\mu} \nabla u_n\nabla v \,\de x $$ be the energy form associated with problem $(\bar P_n)$.

\begin{prop}\label{propexn}
For every given $J\in L^2(\Omega)$ there exists a unique weak solution $u_n\in H^1_0(\Omega_n)$ of problem $(\bar P_n)$.
\end{prop}

The following result states the convergence of the pre-fractal solutions $u_n$ to the solution $u$ of problem $(\bar P)$ in a suitable sense. We recall that, for any compact subset $E\subset \Omega$, its relative capacity with respect to $\Omega$ is defined by
\[\capac_{2,\Omega}(E)=\inf\left\lbrace \left\|\varphi\right\|_{H^1(\Omega)}^2: \text{$\varphi\in C_c^\infty(\Omega)$ and $\varphi\geq 1$ on $E$}\right\rbrace,\]
see \cite[p. 531]{mosco1}.

\begin{theorem}\label{convergenza} Let $u$ and $u_n$ be the solutions of problems $(\bar P)$ and $(\bar P_n)$ respectively. Then $u_n$ strongly converges to $u$ in $H^1_0(\Omega)$ as $n\rightarrow \infty$.
\end{theorem}

\begin{proof} The result follows from \cite{mosco1} since $\Omega_n$ is an increasing sequence of sets invading $\Omega$ and $\capac_{2,\Omega}(\Omega^{\prime} \setminus \Omega_n)\rightarrow 0$ when $n\rightarrow \infty$ for any compact subset $\Omega^\prime$ of $\Omega$.
\end{proof}

\section{Numerical approximation in 2D}
\setcounter{equation}{0}

In this section we perform a numerical approximation of problem $(P)$ by a finite element method. For the sake of simplicity, we put $\mu=1$. Hence problem $(\bar P_n)$ reduces to the following form:
\begin{equation}\label{sistema 2d-h-1}
(\tilde P_n)\begin{cases}
-\Delta u_n=J\quad &\text{in}\,\, \Omega_n,\\[2mm]
u_n=0 &\text{on}\,\,\partial\Omega_n.
\end{cases}
\end{equation}

In order to obtain the optimal rate of convergence of the numerical scheme, we use the theory of regularity in weighted Sobolev spaces developed by Grisvard. Let us introduce the weighted Sobolev space
$$H^2_\eta({\Omega_n})= \{v\in H^1(\Omega_n)\,:\,r^{\eta}\,D^{\beta}v\in L^2(\Omega_n),\,|\beta|=2\},$$

where $r=r(x)$ is the distance from the vertices of $\partial\Omega_n$ whose angles are "reentrant". This space is endowed with the norm 
\begin{center}
$\displaystyle\|u\|_{H^2_\eta
(\Omega)}=\left(\|u\|^2_{H^1(\Omega)}+\sum_{|\beta|=2} \int_{\Omega}
r^{2\eta} |D^{\beta} u(x)|^2\,\de x\right)^{\frac{1}{2}}$.
\end{center}

From Kondrat'ev results {\cite{kond}, \cite[Proposition 4.15]{jerison}} and Sobolev embedding theorem we deduce the following.
\begin{theorem}\label{pesata} Let $u_n$ be the weak solution of problem $(\tilde P_n)$. Then $u_n\in H^2_\eta(\Omega_n)$ for $\eta>\frac{1}{4}$. Moreover, $u_n\in H^s(\Omega_n)$ for $s<\frac{7}{4}$ and $u_n\in C^{0,\delta}(\overline\Omega_n)$ for $\delta=\frac{3}{4}-\varepsilon$ for every $\varepsilon>0$.
\end{theorem}
\medskip

We point out that $u_n\notin H^2(\Omega_n)$ since it has a singular behavior in small neighborhoods of the reentrant corners of $\partial\Omega_n$. Hence we have to construct a suitable mesh compliant with the so-called \emph{Grisvard conditions} \cite{grisvard} in order to obtain the optimal rate of convergence. We refer to \cite{Ce-Dl-La} and \cite{cefalolancia}, where such mesh algorithm was developed (see \cite{Ce-La-Hd} for the case of fractal mixtures). We point out that this mesh algorithm produces a sequence of nested refinements.

The mesh refinement process generates a \emph{conformal} and \emph{regular} family of triangulations $\{T_{n,h}\}$, where $h=\max\{\diam(S),\,S\in T_{n,h}\}$ is the size of the triangulation, which is also compliant with the Grisvard conditions (see Section 5 in \cite{cefalolancia} for the case of interest). We define the finite dimensional space of piecewise linear functions
\begin{center}
$X_{n,h}:=\{v\in C^0 (\overline{\Omega_n})\,:\,v|_{\mathcal{T}}\in\mathbb{P}_1\,\,\forall\,\mathcal{T}\in T_{n,h}\}$.
\end{center}
We set $V_{n,h}:=X_{n,h}\cap H^1_0 (\Omega_n)$. Hence $V_{n,h}$ is a finite dimensional space of dimension $N_h=\{$number of inner nodes of $T_{n,h}\}$. The discrete approximation problem is the following:\\
given $J\in L^2 (\Omega_n)$, find $u_{n,h}\in V_{n,h}$ such that
\begin{equation}\label{semidiscreto}
(\nabla u_{n,h},\nabla v_h)_{L^2(\Omega_n)}=(J,v_h)_{L^2(\Omega_n)}\quad\forall\,v_h\in V_{n,h}.
\end{equation}
\noindent The existence and uniqueness of the semi-discrete solution $u_{n,h}\in V_{n,h}$ of the variational problem \eqref{semidiscreto} follows from the Lax-Milgram theorem (see e.g. \cite{quarval}). 

\begin{theorem}\label{stimaerrore} Let $u_n$ be the solution of problem $(\tilde P_n)$ and $u_{n,h}$ be the solution of the discrete problem \eqref{semidiscreto}. Then
\begin{equation}\label{fordis}
\|u_n-u_{n,h}\|_{H^1_0(\Omega_n)}^2\leq C\,h^2\|J\|^2_{L^2(\Omega_n)},
\end{equation}
where $C$ is a suitable constant independent of $h$.
\end{theorem}

\medskip

{For the proof, see Theorem 8.4.1.6 in \cite{grisvard}.}

We now show some numerical simulations for problem $(\tilde P_n)$. We choose the source $J$ as follows:
\begin{center}
$J(x_1,x_2) = 10^5 \ e^{-{5} ((x_1-\bar{x}_1)^2+(x_2-\bar{x}_2)^2)}$,
\end{center}
where $(\bar{x}_1,\bar{x}_2)$ are the center coordinates of the domain.

\begin{figure}[h]
    \centerline{\includegraphics[scale=.5]{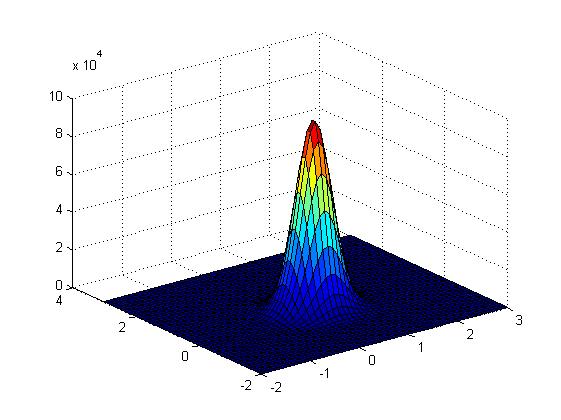}}
	\caption{The source $J$.}
\end{figure}

In our simulations, $\Omega_0$ is the circle of radius $\frac{1}{2}$, while $\Omega_n$, $n=1,\dots,5$, are the domains having as boundary the $n$-th approximation of the Koch snowflake. We suppose that all the domains are centered at the same point.

Denoting by $u_n$ the solution of problem $(\tilde P_n)$, we define the vector ${\mathbf u}_n=(0,0,u_n)$ and we compute the magnetic field ${\mathbf B}$ generated by the current ${\mathbf J}:=(0,0,J(x_1,x_2))$. In other words, ${\mathbf B}=\curl {\mathbf u}_n=\nabla\times {\mathbf u}_n$.

In the following table (Table \ref{table:valori}), we write in the second column the value of the $L^\infty$-norm of ${\mathbf B}$ in $\Omega_n$, while in the third column we write the length $\ell(n)$ of the boundary $\partial\Omega_n$. In the first column, we write the domain we consider in the simulation. 

As one can notice from Table \ref{table:valori}, the magnetic field increases as the length of the boundary of the domain increases.

\medskip

\begin{table}[H]
\begin{center}
\begin{tabular}{ c | c | c | c | c }
  $\Omega_n$ & $\|{\mathbf B}\|_\infty$ & $\ell(n)$ \\
  \hline			
  $\Omega_0$ & 17.946 & $\pi$\\
	$\Omega_1$ & 26.688 & 4\\
	$\Omega_2$ & 35.575 & $\frac{16}{3}$\\
	$\Omega_3$ & 47.124 & $\frac{64}{9}$\\
	$\Omega_4$ & 63.504 & $\frac{256}{27}$\\
	$\Omega_5$ & 85.43 & $\frac{1024}{81}$
\end{tabular}
\caption{The values obtained in our simulations.}
\label{table:valori}
\end{center}
\end{table}

\begin{figure}[p]
\centering
	\includegraphics[scale=0.70]{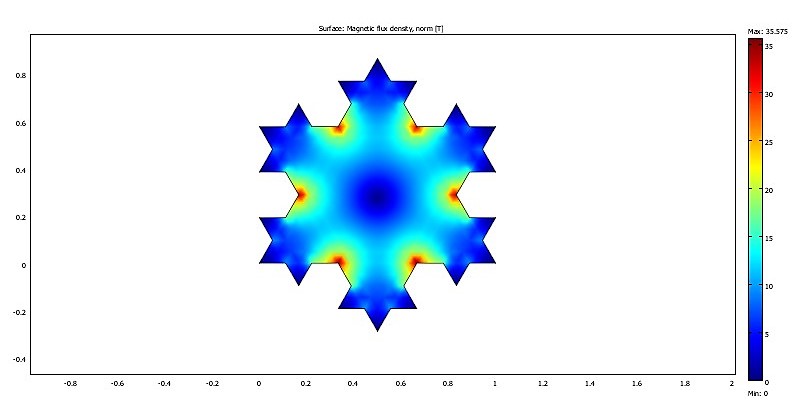}
	\includegraphics[scale=0.70]{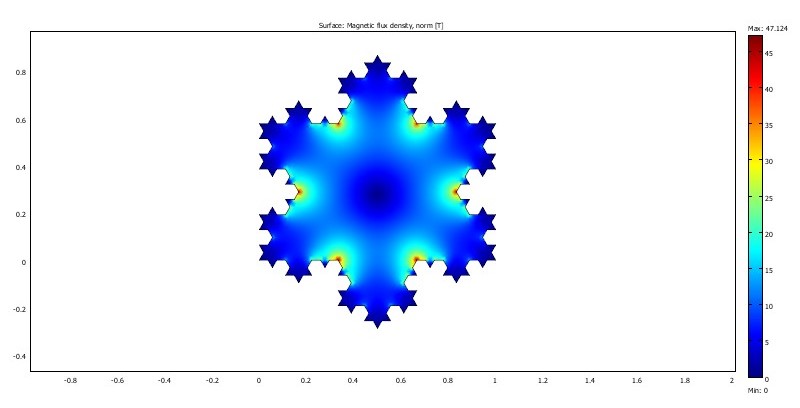}
	\includegraphics[scale=0.70]{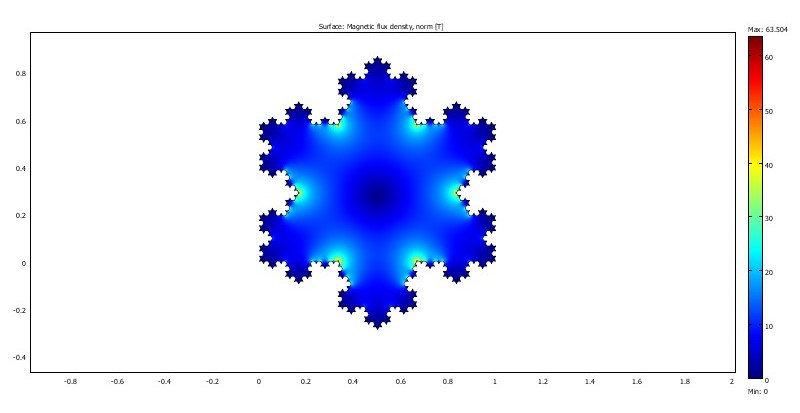}
	\caption{The magnetic field generated by $J$ in $\Omega_2$, $\Omega_3$, $\Omega_4$.}\label{fig-key}
\end{figure}


\begin{remark} 
{We  
note that our numerical results (see Figure \ref{fig-key}) compare well with numerical results 
of Lapidus et al \cite{DL07,GL97,LNRG96} on eigenfunctions of the scalar Dirichlet Laplacian in 
the Koch snowflake domain, and with some earlier physics results, such as  
\cite{sapoval1991vibrations}. 
In particular, one can expect that the localization and other properties of the electro-magnetic fields can be analyzed using similar methods as for the scalar Laplacian (see, for instance, \cite{filoche2012universal,LP95,Lap91,MV97,vdB00b,vdB00a}). 
This connection lies outside of the scope of our article and will be the subject of future research. }
\end{remark}

\noindent {\bf Acknowledgements.} S. C., M. R. L. and P. V. have been supported by the Gruppo Nazionale per l'Analisi Matematica, la Probabilit\`a e le loro Applicazioni (GNAMPA) of the Istituto Nazionale di Alta Matematica (INdAM). M. H. has been supported by the DFG IRTG 2235.
A. T. has been supported by the NSF DMS     1613025.  
The authors  thank the anonymous referee of this paper for the insightful and helpful comments and suggestions on its earlier version.


\end{document}